\documentclass[letterpaper, 10 pt, conference]{ieeeconf}  

\IEEEoverridecommandlockouts                              

\usepackage[utf8]{inputenc} 
\usepackage[T1]{fontenc}    
\usepackage{url}            
\usepackage{booktabs}       
\usepackage{amsfonts}       
\usepackage{nicefrac}       
\usepackage{microtype}      
\usepackage[separate-uncertainty=true,detect-mode=true]{siunitx} 
\usepackage{mathtools}
\usepackage{amssymb}
\usepackage{amsmath}
\usepackage{tabu}
\usepackage{color}
\usepackage{algorithm}
\usepackage{algorithmic}
\usepackage{hyperref}

\usepackage[capitalize]{cleveref}

\newcommand{\expec}[2]{\mathbb{E}_{#2} \brock{#1}
}
\newcommand{\paren}[1]{\mathopen{}\mathclose\bgroup\left(#1\aftergroup\egroup\right)}
\newcommand{\brock}[1]{\mathopen{}\mathclose\bgroup\left[#1\aftergroup\egroup\right]}
\newcommand{\curly}[1]{\mathopen{}\mathclose\bgroup\left\{#1\aftergroup\egroup\right\}}

\usepackage{amsmath,amssymb,mathtools}

\newtheorem{theorem}{Theorem}[section]

\newcommand\bbP{\ensuremath{\mathbb{P}}} 

\newcommand{\suchthat}{\;\ifnum\currentgrouptype=16 \middle\fi|\;}

\DeclarePairedDelimiterX{\infdivx}[2]{(}{)}{%
  #1\;\delimsize\|\;#2%
}

\DeclareMathOperator*{\supp}{supp}

\title{\LARGE \bf
Simulating Emergent Properties of Human Driving 
Behavior \\ Using Multi-Agent Reward Augmented Imitation Learning 
}

\author{Raunak P. Bhattacharyya, Derek J. Phillips, Changliu Liu, \\ Jayesh K. Gupta, Katherine Driggs-Campbell, and Mykel J. Kochenderfer
\thanks{R. Bhattacharyya, D. Phillips, C. Liu, J.K. Gupta, and M.J. Kochenderfer are with the Stanford Intelligent Systems Laboratory in the Department of Aeronautics and Astronautics at Stanford University, Stanford, CA 94305, USA (email: \{raunakbh, djp42, changliuliu, jayeshkg, mykel\}@stanford.edu\}). 
K. Driggs-Campbell is with the Department of Electrical and Computer Engineering at the University of Illinois at Urbana-Champaign, Urbana, IL 61801 (email: krdc@illinois.edu).}%
}

\begin{document}
\maketitle
\thispagestyle{empty}
\bibliographystyle{IEEEtran.bst}
\pagestyle{empty}

\begin{abstract}
Recent developments in multi-agent imitation learning have shown promising results for modeling the behavior of human drivers.
However, it is challenging to capture emergent traffic behaviors that are observed in real-world datasets.  
Such behaviors arise due to the many local interactions between agents that are not commonly accounted for in imitation learning.
This paper proposes Reward Augmented Imitation Learning (RAIL), which integrates reward augmentation into the multi-agent imitation learning framework and allows the designer 
to specify prior knowledge in a principled fashion.
We prove that convergence guarantees for the imitation learning process are preserved under the application of reward augmentation.
This method is validated in a driving scenario, where an entire traffic scene is controlled by driving policies learned using our proposed algorithm.
Further, we demonstrate improved performance in comparison to traditional imitation learning algorithms both in terms of the local actions of a single agent and the behavior of emergent properties in complex, multi-agent settings.

\end{abstract}

\section{Introduction}
Robot learning from human demonstrations has been a subject of significant interest in recent years~\cite{argall2009survey}. 
Imitation learning has been applied to vehicle navigation, humanoid robots, and computer games~\cite{hussein2017imitation}.
This paper focuses on imitation learning for building reliable human driver models.

The autonomous driving literature has established that it is infeasible to build a statistically significant case for the safety of a system solely through real-world testing~\cite{iso201126262,koopman2016challenges}. 
Validation through simulation is an alternative to real-world testing, with the ability to evaluate vehicle performance in large numbers of scenes quickly, safely, and economically~\cite{Morton2018}.
This paper seeks to extend state of the art imitation learning to improve our ability to accurately generate realistic driving scenarios. In such safety critical settings, representative models of human driving behavior are essential in the validation of autonomous driving systems.

\begin{figure*}[t!]
\centering
\resizebox{2\columnwidth}{!}{

\includegraphics[clip]{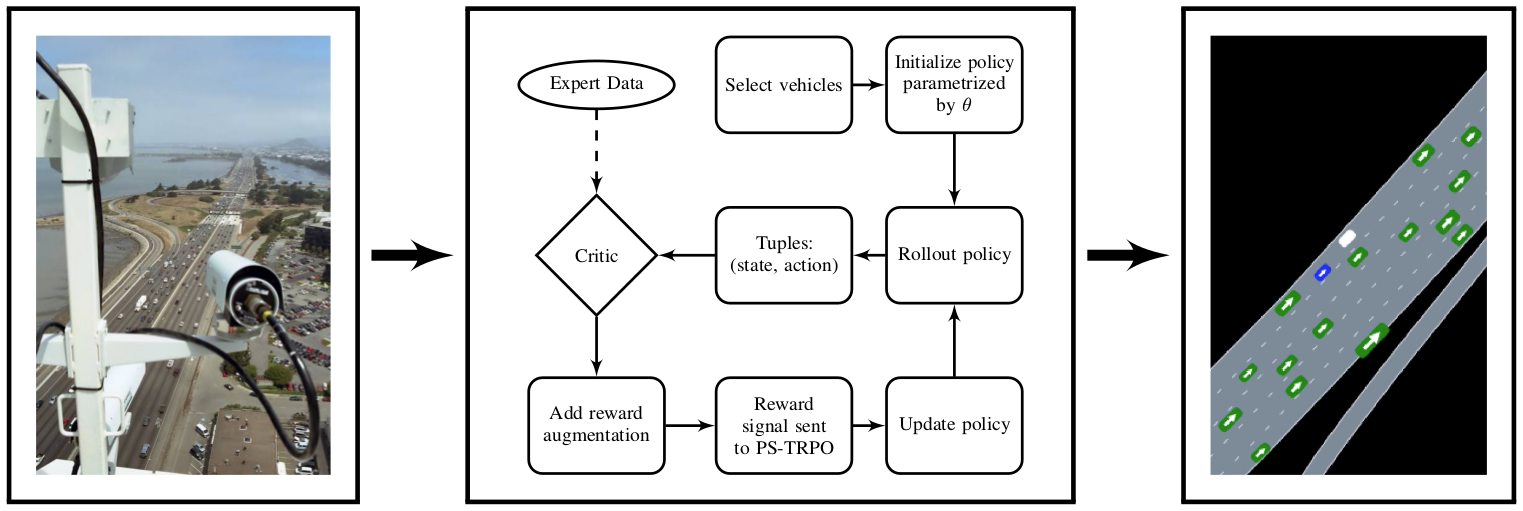}
}

\caption{
\small
The imitation learning pipeline: the demonstration data from the NGSIM dataset~\cite{systematics2007ngsim} (left panel) is fed into RAIL (middle panel) to create driving policies that can be used for validation of autonomous vehicles in simulation (right panel).}
\end{figure*}

Human driving situations are inherently multi-agent in nature. Typical human driving scenes are composed of several vehicles that interact to exhibit emergent patterns of traffic behavior that cannot be easily predicted from the properties of the individual vehicles alone.
For example, given two very similar initial scenes, the vehicles can reach very different configurations after just a few seconds because small changes in states can quickly compound into large differences in the resulting vehicle behaviors and motion
pattern.
Reliable human driver models must be capable of imitating these emergent properties of traffic behavior. 



Generative Adversarial Imitation Learning (GAIL) \cite{ho2016generative} has recently been used to model human driving behavior.
GAIL performs well when imitating the driving behavior in single agent settings~\cite{kuefler2017imitating}, outperforming behavioral cloning and rule-based driver models (e.g., IDM+MOBIL~\cite{treiber2000congested}).
However, GAIL does not scale to imitating the behavior of multiple human driven vehicles because the multi-agent setting leads to the problem of covariate-shift~\cite{bhattacharyya2018multi}.

The PS-GAIL algorithm uses the framework provided by recent work extending GAIL to the multi-agent setting~\cite{song2018magail}, borrowing ideas from the PS-TRPO algorithm~\cite{Gupta2017}.
PS-GAIL outperforms GAIL in terms of imitation performance and scalability in imitating multiple interacting human driven vehicles~\cite{bhattacharyya2018multi}.
However, PS-GAIL has significant room for improvement in terms of imitation performance, specifically in terms of reducing undesirable traffic phenomena arising out of interactions between vehicles, such as off-road driving, collisions, and hard braking.

This paper modifies the GAIL approach to enable the use of external rewards in the multi-agent setting.
The resulting algorithmic procedure improves imitation performance compared to PS-GAIL in three ways. First, it performs better at reproducing individual driving behaviors. Second, resulting policies exhibit reductions in undesirable traffic phenomena such as collisions and offroad driving. Third, emergent properties of multi-agent driving, such as lane changes and spacing between vehicles, are shown to approach human driving behavior. 

While using shaped rewards in the context of imitation learning has been proposed previously~\cite{li2017infogail,judah2014imitation,syed2012imitation,knox2012reinforcement}, this paper differs from the existing literature in two aspects.
First, we propose directly considering the imitation learning problem in the multi-agent setting using parameter sharing.
Second, this approach penalizes undesirable traffic phenomena through reward augmentation, and assesses the resulting impact on emergent properties of multi-agent driving behavior.

This paper makes the following contributions:
(1) We propose a mathematical formulation for incorporating metrics of undesirable traffic phenomena as constraints in the problem of finding policies through imitation learning;
(2) We provide a framework for the designer of the imitation learning agent to provide prior knowledge in the form of reward augmentation to help the learning process; and 
(3) We demonstrate the RAIL algorithm in a case study of modeling human driving behavior and compare the imitation performance against results from existing algorithms. 

The link to the github repository containing the source code for the experiments can be found at 
\url{https://github.com/sisl/ngsim\_env}.

\section{Problem Definition}
We consider the problem of imitation learning in a multi-agent setting. 
The objective is to improve imitation performance by expanding the scope of the imitation algorithm to include imitation of emergent properties.
We hypothesize that by imitating both the local and emergent behaviors, the resulting policy will improve our ability to mimic human behaviors in a multi-agent setting.

\subsection{Formulation}
We formulate highway driving as a sequential decision making task, in which the driver obeys a stochastic policy that maps observed road conditions to a probability distribution over driving actions~\cite{kaelbling1998planning,kochenderfer2015decision}. Given a dataset consisting of a sequence of state-action tuples $(s_t, a_t)$ demonstrating highway driving and a class of policies $\pi_\theta$ parameterized by $\theta$, we adopt imitation learning to infer this policy. 

We use the multi-agent extension of Markov decision processes adapted to the imitation learning framework~\cite{song2018magail,littman1994markov}.
Suppose there are $n$ agents. The state, action, and policy of agent $i$ are denoted $s_i$, $a_i$, and $\pi_i$, respectively. 
The state, action, and policy of the multi-agent system are denoted $\mathbf{s}=[s_1,\ldots,s_n]$, $\mathbf{a}=[a_1,\ldots,a_n]$, and $\bar\pi(s_1, \ldots, s_n)=(\pi_1(s_1), \ldots, \pi_n(s_n))$. 
The state space and the action space of the multi-agent system are denoted $\mathcal{S}$ and $\mathcal{A}$, respectively. 
In the remainder of this paper, we use $s$ and $a$ without the subscripts to refer to the single agent scenario.
We make some simplifying assumptions to the general Markov Games framework, which include agents being homogeneous (every agent has the same action and observation space), each agent getting independent rewards (as opposed to there being a joint reward function), and the reward function being the same for all the agents, as motivated in~\cite{bhattacharyya2018multi}.

We define the $\gamma$-discounted \emph{state occupancy measure} of a policy $\bar\pi$ by $\rho_{\bar\pi}(\mathbf{s}) = \sum_{t=0}^\infty \gamma^t  \bbP_{p_0, \bar\pi}\left[ \mathbf{s}_t = \mathbf{s} \right]$, 
where $\bbP_{p_0, \bar\pi}\left[ \mathbf{s}_t = \mathbf{s} \right]$ is the probability of landing in state $\mathbf{s}$ at time $t$, when following $\bar\pi$ starting from $\mathbf{s_0} \sim p_0$. 
When convenient, we will overload notation for \emph{state-action occupancy measure}: $\rho_{\bar\pi}(\mathbf{s},\mathbf{a}) = \bbP_{\bar\pi}\left[\mathbf{a} \mid \mathbf{s}\right]\rho_{\bar\pi}(\mathbf{s})$, where $\bbP_{\bar\pi}\left[\mathbf{a} \mid \mathbf{s}\right]$ is the probability of executing action $\mathbf{a}$ in state $\mathbf{s}$. 
We denote the \emph{support} of the occupancy measure as $\supp$ where $\supp(\rho_{\bar\pi}):=\{(\mathbf{s},\mathbf{a}):\rho_{\bar\pi}(\mathbf{s},\mathbf{a})>0\}$.

Consider a multi-agent policy $\bar \pi$, which maps the multi-agent system state in $\mathcal{S}$ to a distribution over the multi-agent action space $\mathcal{A}$. 
Given the demonstrated data $\bar \pi_E$, we need to ensure that the greatest difference between $\bar\pi$ and $\bar\pi_E$ is small. The difference is measured between the demonstrated trajectory and the roll-out trajectory given $\bar\pi$ in a finite time horizon.
Using the GAIL framework, our goal is to minimize the distance between the occupancy measures $\rho_{\bar\pi}$ and $\rho_{\bar\pi_E}$.
Further, to ensure centralized training with decentralized control, and to provide prior knowledge, we introduce parameter sharing~\cite{Gupta2017} and reward augmentation~\cite{ng1999policy}.

Mathematically, introducing parameter-sharing and reward augmentation poses two constraints on the function space of the policy $\bar\pi$. 
For parameter-sharing, we are enforcing that $\pi_i = \pi_j = \pi$ for any $i$ and $j$, where $\pi$ denotes the policy for single agent. 
Hence, $\bar\pi(s_1,\ldots,s_n) = (\pi(s_1),\ldots,\pi(s_n))$. 
For reward augmentation, we require that $\pi$ belongs to a certain set such that undesired actions are discouraged. 
For example, the vehicle should not drive off road, collide with others, or brake too hard.
Such undesired state-action pairs are denoted as belonging to the set $U$. 
The constraint on the policy is denoted by $\Pi$:
\begin{equation}
\label{eq: constraint on policy}
\Pi = \{\bar\pi : \pi = \pi_i, \forall \, i \text{, and } \bbP_{\bar\pi}\left[\mathbf{a} \mid \mathbf{s}\right] = 0,\forall \, (\mathbf{s},\mathbf{a})\in U\} 
\end{equation}

Considering Wasserstein distance~\cite{arjovsky2017wasserstein}, the following constrained minimax problem for imitation learning is formulated:
\begin{equation}\label{eq: general constrained minimax formulation}
\min_{\bar\pi\in\Pi}  \max_{D} \left\{\expec{D(\mathbf{s},\mathbf{a})}{\bar\pi_E}-\expec{D(\mathbf{s},\mathbf{a})}{\bar \pi}\right\} 
\end{equation}
where the critic, $D$, learns to output a high score when encountering pairs from $\bar\pi_{E}$, and a low score when encountering pairs from $\bar\pi$. $D$ should be optimized for all functions.

\subsection{Solution Approach}
The constrained minimax is solved by transforming the problem to an unconstrained form.
The constraint for parameter sharing is naturally encoded by sharing the same policy for all agents.
The constraint for reward augmentation is enforced by adding a reward augmentation regularizer in the function. 
Thus, the unconstrained problem becomes:
\begin{equation}
\underset{\pi}{\min} \
\underset{D}{\max} \
\expec{D(\mathbf{s},\mathbf{a})}{\bar\pi_{E}} 
- \expec{D(\mathbf{s},\mathbf{a})}{\pi} + r\expec{\mathbf{1}_U}{\pi} 
\end{equation}
where $r$ is the penalty, and $\mathbf{1}_U$ is an indicator function that is non-zero if and only if $(\mathbf{s},\mathbf{a})\in U$. 
The penalty $r$ can either be a constant value or a barrier function.
We have binary penalty when $r$ is constant, and smooth penalty when $r$ is a continuous function that reaches zero on the boundary of the set $U$. 
Note that the term $\expec{D(\mathbf{s},\mathbf{a})}{\pi}$ is different from $\expec{D(\mathbf{s},\mathbf{a})}{\bar\pi}$, where the former notation requires that all agents use the same policy $\pi$ with shared parameters.

We are now ready to introduce an algorithm to solve the problem as formulated above.
This algorithm is called Reward Augmented Imitation Learning (RAIL, cf. Risk Averse Imitation Learning~\cite{santara2017rail}).
We parameterize the single agent policy using $\theta$ and the critic using $\psi$. The parameterized policy and critic are denoted by $\pi_\theta$ and $D_\psi$, respectively.
Under this parametrization, the objective function becomes:
\begin{equation}
\label{eq:wasserstein}
\begin{aligned}
\underset{\theta}{\min} \
\underset{\psi}{\max} \
\expec{D_{\psi}(\mathbf{s},\mathbf{a})}{\bar\pi_{E}} 
- \expec{D_{\psi}(\mathbf{s},\mathbf{a})}{\pi_{\theta}} + r\expec{\mathbf{1}_U}{\pi_\theta} \text{.}
\end{aligned}
\end{equation}

To solve for the desired $\pi_\theta$, the following two steps are performed iteratively:

\textsc{Step 1:} maximize $D_{\psi}$.
Similar to single agent GAIL, this step involves the rollout and the update of the critic $D_{\psi}$.

\textsc{Step 2:} minimize policy $\pi_\theta$.
This step is where the constraints are taken into account.


\cref{algo:rail} provides the pseudo code that enacts the above two step procedure and incorporates the reward augmentation by providing penalties. 



\begin{algorithm}[t]
  \caption{RAIL}
  \label{algo:rail}
  \begin{algorithmic}
    \STATE {\bfseries Input:} Expert trajectories $\tau_E \sim \bar\pi_E$, Shared policy parameters $\Theta_0$, Critic parameters $\psi_0$, Trust region size $\Delta_{KL}$
    \FOR{$k \gets 0, 1, \dotsc$}
    \STATE Rollout trajectories for all agents $\vec{\tau} \sim \pi_{\theta_k}$
    \STATE Score $\vec{\tau}$ with critic, generating reward $\tilde{p}(s_{t},a_{t};\psi_k)$ and \\ $\quad$ added penalty $r$
    \STATE Batch trajectories obtained from all the agents
    \STATE Calculate advantage values
    \STATE Take a Trust Region Policy Optimization (TRPO)~\cite{schulman2015trust} step to find $\pi_{\theta_{k+1}}$ using $\Delta_{KL}$
    \STATE Update the critic parameters $\psi$ 
    \ENDFOR
  \end{algorithmic}
\end{algorithm}

\vspace{-10pt}
\subsection{Theoretical Analysis: Convergence and Optimality}
This section shows the convergence and optimality of the constrained minimax problem using non-parameterized policy and critic. 
In the following discussion, the occupancy measure refers to the \emph{state-action occupancy measure}. 
Given the occupancy measure $\rho_{\bar\pi}$,  we have
\begin{eqnarray}
\expec{D(\mathbf{s},\mathbf{a})}{\bar\pi} &=& \sum_{t=0}^\infty\gamma^t D(\mathbf{s}_t,\mathbf{a}_t) \nonumber\\
&=& \int_{\mathcal{A}}\int_{\mathcal{S}}\rho_{\bar\pi}(\mathbf{s},\mathbf{a})D(\mathbf{s},\mathbf{a})d\mathbf{s}d\mathbf{a}\text{,}
\end{eqnarray}
where $\mathbf{s}_t$ and $\mathbf{a}_t$ are rollout data from the policy $\bar\pi$. 
By Proposition 3.1 of \cite{ho2016generative}, there is a one-to-one correspondence between the policy $\bar\pi$ and the occupancy measure $\rho_{\bar\pi}$. 

The minimax objective function \eqref{eq: general constrained minimax formulation} can be re-written as
\begin{eqnarray}
&\min_{\rho_{\bar\pi}\in\Sigma} \quad \max_{D}  &\int_{\mathcal{A}}\int_{\mathcal{S}}\rho_{\bar\pi_E}(\mathbf{s},\mathbf{a})D(\mathbf{s},\mathbf{a})d\mathbf{s}d\mathbf{a} \nonumber\\
&&-\int_{\mathcal{A}}\int_{\mathcal{S}}\rho_{\bar\pi}(\mathbf{s},\mathbf{a})D(\mathbf{s},\mathbf{a})d\mathbf{s}d\mathbf{a} 
\label{eq: minimax occupancy measure}
\end{eqnarray}
where $\Sigma$ is the constraint on the occupancy measure equivalent to \eqref{eq: constraint on policy} with
\begin{eqnarray}
&\Sigma = \{ \rho_{\bar\pi}:& \rho_{\pi}(s_i,a_i) =  \rho_{\pi}(s_j,a_j), \forall i,j\text{; }\nonumber\\ && \rho_{\bar\pi}(\mathbf{s},\mathbf{a})=0,\forall (\mathbf{s},\mathbf{a})\in U\}\text{,}\label{eq: constraint minimax constraint}
\end{eqnarray}
and $\rho_{\pi}(s_i,a_i)$ is a marginal occupancy measure by integrating out $(s_j,a_j)$ for $j\neq i$ in $\rho_{\bar\pi}(\mathbf{s},\mathbf{a})$.

\begin{theorem}\label{thm: convergence and optimality}
The solution converges in measure by iteratively solving the minimax problem \eqref{eq: minimax occupancy measure} if the following conditions hold:\footnote{In practice, these assumptions may not exactly hold.}
\begin{enumerate}
\item $D$ and $\bar\pi$ have enough capacity (representing all functions);
\item $D(\mathbf{s},\mathbf{a})$ attains the optimal value $D^*(\mathbf{s},\mathbf{a})$ at each iteration;
\item $\rho_{\bar\pi}$ is updated so as to improve the criterion $\min_{\rho_{\bar\pi}\in\Sigma} \int_{\mathcal{A}}\int_{\mathcal{S}}\left[\rho_{\bar\pi_E}(\mathbf{s},\mathbf{a})-\rho_{\bar\pi}(\mathbf{s},\mathbf{a})\right]D^*(\mathbf{s},\mathbf{a})d\mathbf{s}d\mathbf{a}$.
\end{enumerate}
Moreover, the solution converges in measure to the optimal solution of the following problem:
\begin{equation}
\min_{\rho_{\bar\pi}\in\Sigma} \int_{\mathcal{A}}\int_{\mathcal{S}}\|\rho_{\bar\pi_E}(\mathbf{s},\mathbf{a})-\rho_{\bar\pi}(\mathbf{s},\mathbf{a})\|^2d\mathbf{s}d\mathbf{a}\text{.}\label{eq: objective occupancy measure two norm}
\end{equation}
\end{theorem}

\begin{proof}
If \eqref{eq: minimax occupancy measure} is unconstrained, i.e., $\Sigma$ is the whole occupancy measure space, ~\cite{goodfellow2014generative} and~\cite{arjovsky2017wasserstein} have shown that the solution converges by iteratively solving the unconstrained minimax problem if the three conditions hold.

Consider the constrained version of the optimization. 
Since both the objective function \eqref{eq: minimax occupancy measure} and the constraint \eqref{eq: constraint minimax constraint} are convex in $\rho_{\bar\pi}$, the constraints will not affect convergence, as long as the three conditions are satisfied.\footnote{We are dealing with the function space which contains $\rho_{\bar\pi}$. The convexity means: for any $\rho_1$ and $\rho_2$ satisfying \eqref{eq: constraint minimax constraint}, their linear interpolation $(1-\lambda)\rho_1+\lambda \rho_2$ also satisfy \eqref{eq: constraint minimax constraint} for any $\lambda\in[0,1]$. It is easy to verify that this condition is true for any $U$, i.e., if $\rho_1$ and $\rho_2$ are $0$ on $U$, then $(1-\lambda)\rho_1+\lambda \rho_2$ should also be $0$ on $U$. Hence, the convexity of the constraint is not affected by the shape of $U$.}  
The solution of the inner maximization in \eqref{eq: minimax occupancy measure} is $D^*(\mathbf{s},\mathbf{a}) \propto \rho_{\bar\pi_E}(\mathbf{s},\mathbf{a})-\rho_{\bar\pi}(\mathbf{s},\mathbf{a})$. Then, in the limit, the problem
\begin{equation}
\min_{\rho_{\bar\pi}\in\Sigma} \int_{\mathcal{A}}\int_{\mathcal{S}}\left[\rho_{\bar\pi_E}(\mathbf{s},\mathbf{a})-\rho_{\bar\pi}(\mathbf{s},\mathbf{a})\right]D^*(\mathbf{s},\mathbf{a})d\mathbf{s}d\mathbf{a}    
\end{equation}
converges to \eqref{eq: objective occupancy measure two norm} in measure.
\end{proof}

According to \cref{thm: convergence and optimality}, if the demonstrated data satisfy the constraint $\Sigma$, then $\rho_{\bar\pi}\rightarrow\rho_{\bar\pi_E}$ in the limit. 
Consider the case that the demonstrated data does not satisfy $\Sigma$. 
The constraint for parameter sharing $\rho_{\pi}(s_i,a_i) =  \rho_{\pi}(s_j,a_j)$  introduces an averaging effect, i.e., the learned single agent policy $\rho_{\pi}$ is an average of the demonstrated single agent policy $\rho_{\pi_{Ei}}(s_i,a_i)$ for all $i$. 
The demonstrated single agent policy $\rho_{\pi_{Ei}}(s_i,a_i)$ is a marginal occupancy measure by integrating out $(s_j,a_j)$ for $j\neq i$ in $\rho_{\bar\pi_E}(\mathbf{s},\mathbf{a})$. 
The average effect results from identity permutation during  minimization of the two norm in \eqref{eq: objective occupancy measure two norm}. 
The constraint for reward augmentation  $\rho_{\bar\pi}(\mathbf{s},\mathbf{a})=0,\forall (\mathbf{s},\mathbf{a})\in U$ introduces a truncation effect, which shrinks the support of $\rho_{\pi}$ such that $\supp(\rho_{\pi})\in\Gamma:=U^c$. 
The truncation effect is illustrated in \cref{fig: constraints}. 

Normally, the constraint from reward augmentation is satisfied in the demonstration data.
For example, the vehicles do not drive off the road.
However, the constraint from parameter sharing may not be satisfied by the demonstrated data, i.e., the assumption that all vehicles are homogeneous may not hold.
Then, as discussed above, the learned policy takes an average over different policies.
Therefore, reward augmentation improves the learning performance in practice, since it encodes prior knowledge.

\begin{figure}[t]
\begin{center}
\resizebox{\columnwidth}{!}{
\includegraphics[]{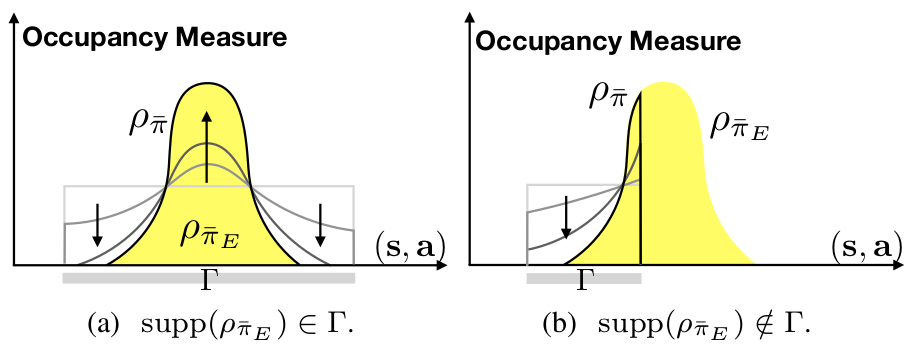}
}
\vspace{-0.5cm}
\caption{\small The effect of constraint $\Gamma$ on convergence and the learned policy. Shaded area represents the occupancy measure of the demonstration $\rho_{\bar\pi_E}$. 
Curves represent the learned occupancy measure $\rho_{\bar\pi}$ at different iterations (lighter colors indicate earlier iterations). 
A uniform distribution over the constraint $\Gamma$ initializes $\rho_{\bar\pi}$.
When $\supp(\rho_{\bar\pi_E})\in\Gamma$, the demonstrated occupancy measure can be recovered by $\rho_{\bar\pi}$. When $\supp(\rho_{\bar\pi_E})\notin\Gamma$, the demonstrated occupancy measure outside $\Gamma$ can not be recovered.}
\label{fig: constraints}
\end{center}
\vspace{-20pt}
\end{figure}

\section{Experiments and Results}
We use the results from PS-GAIL as a baseline to compare against the RAIL algorithm by learning policies and calculating specific metrics, as described in~\cref{metrics}. 
We train three policies for each set of parameters that we want to compare, selecting the best performing policy based on the results of \num{10000} policy rollouts in the \num{100}-agent training environment.
The results presented in~\cref{results} are extracted by evaluating our policies in the same manner, but on scenes sampled from the held-out testing dataset. 

\subsection{Environment}
We evaluate our algorithm using the same simulator as is used in the development of PS-GAIL~\cite{bhattacharyya2018multi}. 
The simulator allows us to sample initial scenes from real traffic data and then simulate for \SI{20}{\second} at \SI{10}{\hertz}. 
The most important feature of this simulator is that expert vehicles observed in the real data can be replaced with policy controlled agents, crucial to both learning a good policy and evaluating final policies.
We replace \num{100} vehicles from the initial scene with vehicles driven by the learned policy.
Another crucial component of the simulator is the extraction of features from the environment which are then fed into the policy controller as observations.
The agent's decisions are translated into actions, which the simulator uses to determine the next state.

We run our experiments on data from the Next Generation Simulation (NGSIM) project~\cite{systematics2007ngsim}.
This dataset is split into three consecutive \SI{15}{\minute} sections of driving data for a fixed section of highway 101 in California.
We use the first section as the training dataset, from which we learn our policies.
The remaining two sections are used for testing and evaluating the quality of the resulting policies. 

\subsection{Reward Augmentation}
Reward augmentation combines imitation learning with reinforcement learning and helps improve state space exploration of the learning agent.
Part of the reinforcement learning reward signal comes from the critic based on imitating the expert, and another signal comes from the externally provided reward specifying the prior knowledge of the expert~\cite{li2017infogail}.

The reward augmentation in our experiments is provided in the form of penalties.

\subsubsection{Binary Penalty}
\label{binaryrewarddetails}
The first method of reward augmentation that we employ is to penalize states in a binary manner, where the penalty is applied when a particular event is triggered.
To calculate the augmented reward, we take the maximum of the individual penalty values.
For example, if a vehicle is driving off the road and colliding with another vehicle, we only penalize the collision. 
This will also be important when we discuss \textit{smoothed penalties}.

We explore penalizing three different behaviors. 
First, we give a large penalty $R$ to each vehicle involved in a collision.
Next, we impose the same large penalty $R$ for a vehicle that drives off the road. 
Finally, performing a hard brake (acceleration of less than \SI{-3}{\meter\per\second^{2}}) is penalized by only $\frac{R}{2}$.
The penalty formula is shown in~\cref{binaryreward}.
We denote the smallest distance from the ego vehicle to any other vehicle on the road as $d_c$ (meters), where $d_c \geq 0$.
We also define the closest distance from the ego vehicle to the edge of the road (meters): $d_{\text{road}} = \min\{d_{\text{left}}, d_{\text{right}}\}$. 
We allow $d_{\text{road}}$ to be negative if the vehicle is off the road. 
Finally, let $a$ be the acceleration of the vehicle, in \SI{}{\meter\per\second^2}. A negative value of $a$ indicates that the vehicle is braking.
Now, we can formally define the binary penalty function:
\begin{align}
\label{binaryreward}
    \text{Penalty} & = 
    \begin{cases}
       R &  d_c = 0 \\
       R & d_{\text{road}} \leq -0.1 \\
       \frac{R}{2} & a \leq -3
    \end{cases} 
\end{align}
The relative values of the penalties indicate the preferences of the designer of the imitation learning agent.
For example, in this case study, we penalize hard braking less than the other undesirable traffic phenomena.

\subsubsection{Smooth Penalty}
In this case, we provide a \textit{smooth penalty} for off-road driving and hard braking, where the penalty is linearly increased from a minimum threshold to the previously defined event threshold for the binary penalty. 

For off-road driving, we linearly increase the penalty from $0$ to $R$ when the vehicle is within \SI{0.5}{\meter} of the edge of the road.
For hard braking, we linearly increase the penalty from $0$ to $R/2$ for acceleration between \SI{-2}{\meter\per\second^{2}} and \SI{-3}{\meter\per\second^{2}}.

\subsection{Metrics}
\label{metrics}
We assess the imitation performance of our driving policies at three levels. These are imitation of local driving behavior, reduction of undesirable traffic phenomena, and imitation of emergent properties of multi-agent driving.
First, to measure imitation of local vehicle behaviors, we use a set of Root Mean Square Error (RMSE) metrics that quantify the divergence between the trajectories generated by our learned policies and the real trajectories in the dataset. 
We calculate the RMSE between the original human driven vehicle and its replacement policy driven vehice in terms of the position, lane offset, and speed.
A perfect policy would have RMSE values close to \num{0} for the entire rollout duration. 

Second, to assess the undesirable traffic phenomena that arise out of vehicular interactions as compared to local, single vehicle imitation, we extract metrics that quantify hard braking, collisions, and offroad driving.
It is important to note that these undesirable traffic phenomena were explicitly incorporated into the formulation of penalty based reward augmentation provided to the RAIL algorithm.
We also extract these metrics of undesirable traffic phenomena for the NGSIM driving data and compare them against the metrics obtained from rollouts generated by our driving policies.

Third, to quantify imitation of driving properties that are emergent in that they are not explicitly modeled in the RAIL formulation, we assess metrics of emergent properties.
These are the average number of lane changes per vehicle, the average timegap per vehicle, and the distribution of speed over all vehicles. 
The timegap for a vehicle is defined as the time spacing (in seconds) to the vehicle in front of it. 
These metrics of emergent driving properties are calculated for the NGSIM driving data and compared against metrics obtained from rollouts generated by our driving policies.

\subsection{Results}
\label{results}
We compare our proposed algorithm, RAIL, against PS-GAIL.
For comparisons between PS-GAIL, traditional GAIL, and rule-based models, we guide the reader to~\cite{kuefler2017imitating,bhattacharyya2018multi}.
The policies generated using RAIL were obtained using $R=2000$ for the binary penalty, and $R=1000$ for the smoothed penalties. 
These $R$ values were determined to be the best after performing a hyperparameter search on penalty values ranging from \num{0} and \num{5000}. 

\begin{figure}[t]
     \centering
     
     \resizebox{\columnwidth}{!}{

\includegraphics[clip]{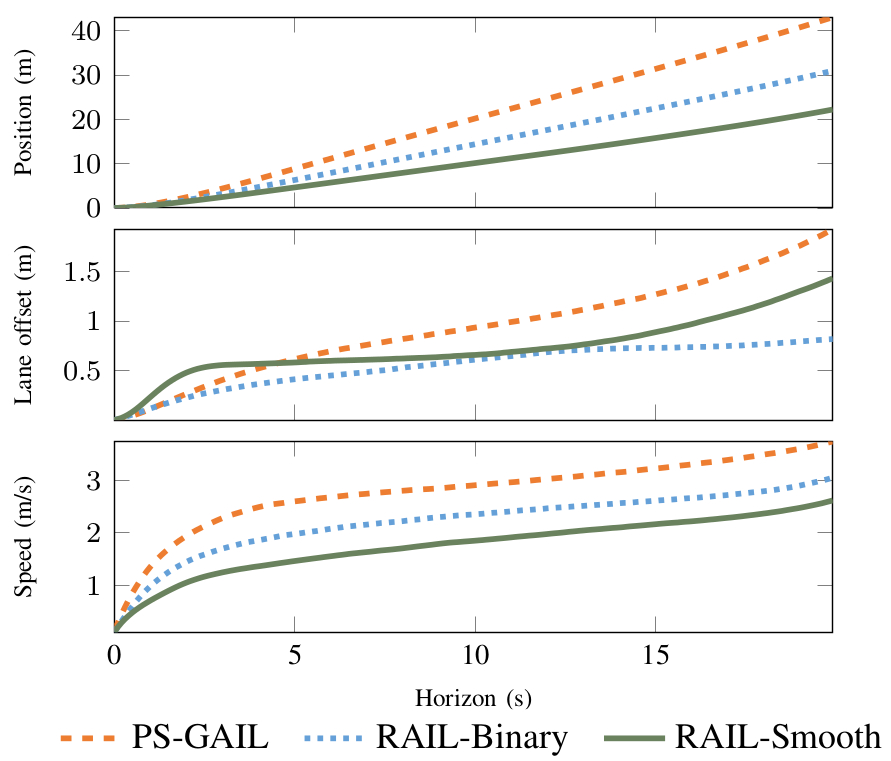}
}
\vspace{-0.5cm}
     \caption{
     \small
     Root mean square error with respect to NGSIM data with increasing time horizon as a measure of local imitation performance. Policies trained using RAIL demonstrate closer driving behavior to human demonstration as compared to the PS-GAIL baseline.
     }
     \label{rmse}
 \end{figure}

\Cref{rmse} shows the RMSE values for speed, lane offset and position of the vehicle driven by imitation learned policies varying with increasing time horizon of the simulation rollout. Policies learned using RAIL show lower values of RMSE as compared to PS-GAIL throughout the rollout duration. Further, between the two RAIL policies, it is observed that smoothing the penalties improves the RMSE performance. Thus, RAIL outperforms PS-GAIL in imitating local driving behavior of individual vehicles. 

\begin{figure}[b]
   \centering
   \resizebox{\columnwidth}{!}{

\includegraphics[]{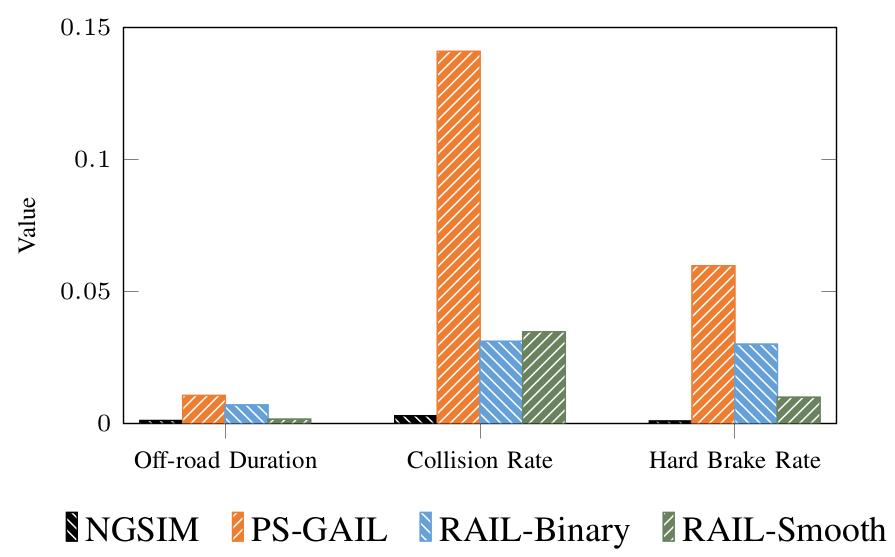}
}
\vspace{-0.5cm}
   \caption{
   \small
   Metrics of undesirable traffic phenomena. These are explicitly penalized in the reward augmentation formulation. RAIL results in policies with lower values of collisions, offroad driving and hard braking as compared to the PS-GAIL baseline.
   }
   \label{undesirablemetrics}
\end{figure}

\Cref{undesirablemetrics} illustrates the number of undesirable traffic phenomena through the metrics of collisions, hard braking, and offroad driving in case of NGSIM data, and policies trained using PS-GAIL and RAIL.
The results show that policies learned using RAIL are less likely to lead vehicles into extreme decelerations, off-road driving, and collisions.
Additionally, for the case where we provide smooth penalties (off-road duration and hard brake), we see significant reductions in the associated metrics as compared to PS-GAIL.

\begin{figure}[]
   \centering
   \resizebox{\columnwidth}{!}{

\includegraphics[]{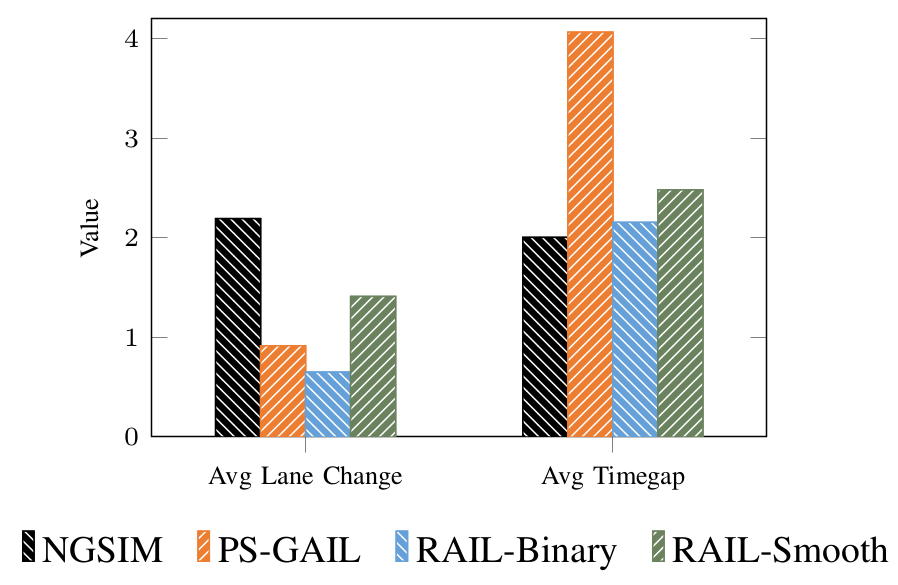}
}
\vspace{-0.5cm}
   \caption{
   \small 
   Metrics of emergent driving behavior. Policies trained using reward augmentation result in lane changing and timegap behavior that is closer to human driving as compared to PS-GAIL.
   }
   \label{emergentmetrics}
\end{figure}

\Cref{emergentmetrics,veldistb} show the imitation performance of policies in terms of emergent properties of driving behavior.
Average number of lane changes per agent, and average timegap per agent are illustrated in~\cref{emergentmetrics}.
While it can be argued that the results showing improvements in reducing undesirable traffic phenomena in~\cref{undesirablemetrics} can be attributed directly to penalizing via reward augmentation, the properties of driving reported in~\cref{emergentmetrics} are truly emergent in that they arise out of vehicular interactions that are not explicitly accounted for in the imitation learning formulation.
Policies trained using reward augmentation result in driving behavior that leads to emergent properties that are closer to human demonstrations as compared to the baseline policies trained using PS-GAIL.

\begin{figure}[b]
   \centering
   \resizebox{\columnwidth}{!}{
\includegraphics[]{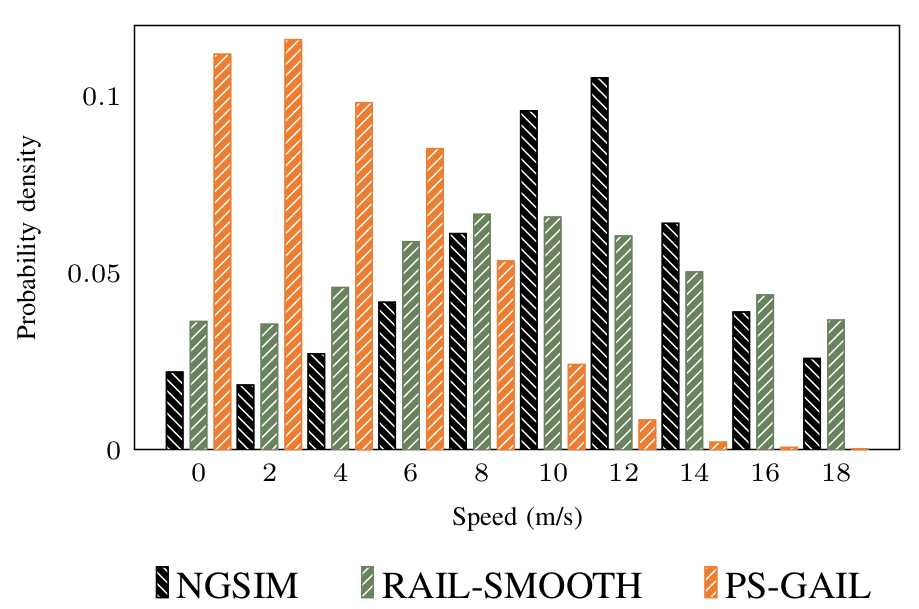}
  }
  \vspace{-0.5cm}
  \caption{
\small
Distribution of speeds over all cars over all the simulation trajectories for human demonstration, policies trained using PS-GAIL, and policies trained using RAIL. RAIL driving policies result in closer velocity distribution to the human demonstrations.}
\label{veldistb}

\end{figure}

The distribution of speeds over all the vehicles in the trajectory is shown in~\cref{veldistb}. The speed values over the trajectory have been normalized and presented as a probability distribution. 
Policies trained using reward augmentation provide speed distributions more closely matching the NGSIM data. Further, the mode of the distribution is closer to human demonstrations in case of RAIL.

\section{Conclusions}
This paper discusses the problem of imitation learning in multi-agent settings in the context of autonomous driving.
The goal of this paper is to create reliable models of human driving behavior that can imitate emergent properties of driving behavior arising out of local vehicular interactions.
Specifically, we provide a framework for multi-agent imitation learning in terms of policy optimization with added constraints using GAIL.
We demonstrate improved performance in learning human driving behavior models as measured by both local and emergent imitation performance.
The main contribution of this paper was including reward augmentation in the imitation learning framework as an added reinforcement signal to the learning agent.
Using externally specified rewards, the designer of the learning agent can provide prior knowledge to guide the training process.
This imitation learning procedure was demonstrated using the RAIL algorithm.

Simulation experiments were performed on learned driving policies from human driving demonstrations in the NGSIM dataset.
These experiments were performed in the multi-agent setting where multiple cars in the NGSIM scenes were replaced by the vehicles driven using polcies learned using RAIL.
Resulting metrics such as root mean square error, off-road duration, collision rates and hard brake rate were used to assess the imitation performance.
The results obtained showed better imitation performance using reward augmentation as compared to previous multi-agent results, especially in terms of imitating emergent properties of driving behavior, as measured by lane changes, timegap and speed distributions of the resulting driving behavior.
Further, this paper also provided theoretical convergence guarantees in the reward augmented imitation learning framework.

A limitation of this approach is that it does not capture different types of driving behavior.
Future work will include latent states to capture different driving styles and enable learning different policies for different agents.
Another interesting extension of this work would focus on populating driving scenarios with these models trained using imitation learning. 
Such scenarios will enable validation of autonomous cars driven using planning algorithms by simulating interactive driving behavior between autonomous vehicles and human driven vehicles.
Finally, policies trained using the RAIL algorithm will be deployed in simulation with an autonomous vehicle for validation testing.

\section*{Acknowledgments}
We thank Blake Wulfe and Jeremy Morton for useful discussions. Toyota Research Institute (TRI) provided funds to assist the authors with their research, but this article solely reflects the opinions and conclusions of its authors and not TRI or any other Toyota entity.


\clearpage
\newpage

\bibliographystyle{IEEEtran}
\bibliography{ref}


\end{document}